\documentclass[preprint,12pt]{elsarticle}

\usepackage{amssymb}
\usepackage{amsthm}
\usepackage{graphicx}
\usepackage{dcolumn}
\usepackage{bm}
\usepackage{setspace}
\usepackage{amssymb}   
\usepackage{color}
\usepackage{amscd}
\usepackage{setstack}
\usepackage{multirow}
\usepackage{ulem}
\usepackage[all]{xy} 
\usepackage{mathrsfs}
\newtheorem{theorem}{Theorem}
\newtheorem{lemma}[theorem]{Lemma}
\newtheorem{definition}[theorem]{Definition}
\newtheorem{corollary}[theorem]{Corollary}

\begin{document}

\begin{frontmatter}



\title{Anosov Properties of a Symplectic Map with Time-Reversal Symmetry}


\author[inst1]{Ken-ichi Okubo}

\affiliation[inst1]{organization={Department of Applied Mathematics and Physics, 
	Graduate School of Informatics, Kyoto University},
            addressline={Yoshida-honmachi}, 
            city={Sakyo-ku},
            postcode={606-8501}, 
            state={Kyoto},
            country={Japan}}
\ead{okubo@rs.socu.ac.jp}
\author[inst1]{Ken Umeno}






\begin{abstract}
This study presents a specific symplectic map, derived from a Hamiltonian, as a model that exhibits time-reversal symmetry on a microscopic scale. Based on the analysis, any initial density function, defined almost everywhere, converges to a uniform distribution in terms of mixing (irreversible behavior) on a macroscopic level. Furthermore, we established that this mixing invariant measure is a unique equilibrium state, unique SRB measure, and physical measure. Additionally, through analytical proof, we have shown that the Kolmogorov-Sinai entropy, representing the average information gain per unit time is positive. This was achieved by validating the Pesin's formula and demonstrating that the critical exponent of the Lyapunov exponent is $1/2$.
\end{abstract}




\end{frontmatter}


\section{Introduction}
\label{sec:introduction}
The relaxation from non-equilibrium to equilibrium states has garnered significant attention since the era of Boltzmann to date.
Boltzmann elucidated the monotonic decrease property of the H-function for classical dynamical systems with time-reversal symmetry by assuming the molecular chaos hypothesis \cite{cercignaniboltzmann,gallavotti1999statistical}.
However, his assertions were met with skepticism. Zermelo challenged Boltzmann’s ideas by invoking Poincaré's recurrence theorem, suggesting that if the state of dynamics is expressed by a single point in phase space, 
trajectories will eventually return arbitrarily close to their initial points, leading to a return of the H-function to its initial value and hence not monotonically decrease.
Boltzmann, in response, argued that the recurrence time is longer than the age of the universe and therefore would be unobservable \cite{gallavotti1999statistical}.
This counterargument by Zermelo stemmed from the notion that the values of the H-function are solely determined by the microscopic state of a single point in phase space.

However, to counter the idea that the microscopic state of a single point in $\Gamma$ space determines macroscopic observation values, an alternative viewpoint emerged, suggesting that 
macroscopic observations are influenced by the mean behavior of points in $\Gamma$ space.
Gibbs introduced the concept of ensembles, focusing on the mean behavior of a set of points on an energy surface \cite{dorfman1999introduction}.
In modern mathematical terms, macroscopic quantities are described as parameters of probability distributions in phase space, known as macroscopic states \cite{keller1998equilibrium,ruelle1999smooth}.

While Boltzmann initially proposed ergodicity to explain the relaxation to equilibrium states, 
phenomena in which systems relax to equilibrium states in relatively short timescales cannot be explained solely through ergodicity \cite{gallavotti1999statistical}.
In such cases, the role of chaos and mixing properties becomes crucial for the expected values of observables of time-dependent probability distribution functions to converge to their values at equilibrium within realistic observation times \cite{gallavotti1999statistical,dorfman1999introduction,zaslavsky1999chaotic,gallavotti2014nonequilibrium}.
Gallavotti proposed the chaos hypothesis, exemplified by mixing Anosov systems \cite{gallavotti1999statistical,gallavotti2014nonequilibrium}.

Mixing involves the spreading out of an initial set of points to occupy the entire phase space while preserving its measure in a coarse sense.
The anticipated value of the time-dependent distribution function $\rho(\Gamma,t)$ of the observed quantity $F(\Gamma)$ at a point $\Gamma$ in phase space is expressed as follows, with $\mu$ representing the invariant measure \cite{dorfman1999introduction}.
\begin{equation}
	\left\langle F(\Gamma)\right\rangle_t = \frac{\int d\mu \rho(\Gamma,t)F(\Gamma)}{\int d\mu \rho(\Gamma,t)}
\end{equation}
Let $\mathcal{M}$ be the entire phase space, $\mathcal{M}$ be partitioned into a sufficiently fine set $\{m_j\}$, and $F_j$ be a typical value of $F$ on $m_j$. Therefore, in the limit as $t \to \infty$, the following equation is satisfied\cite{dorfman1999introduction}:
\begin{equation}
	\lim_{t\to \infty}\left\langle F(\Gamma)\right\rangle_t = \frac{1}{\mu(\mathcal{M})}\sum_j F_j \mu(m_j) = \frac{\int d\mu F(\Gamma)}{\int d\mu}.
\end{equation}
This relaxation to the average value of the equilibrium state is a typical irreversible phenomenon.
Examples of systems exhibiting mixing behavior include Sinai's billiards \cite{sinai1970dynamical}, Bunimovich's stadium \cite{bunimovich1979ergodic}, Wojtkowski's cascade \cite{wojtkowski1986principles},
Arnold's cat map \cite{arnold1968ergodic}, and (multi)baker's map \cite{dorfman1999introduction,tasaki1998analytical,tasaki1995fick}.

Anosov diffeomorphisms \cite{katok1995introduction,bowen1975equilibrium} represent mathematical models of chaotic phenomena. These diffeomorphisms possess hyperbolic structures and adhere to the following three conditions \cite{tasaki1998analytical}.
Before delving into these conditions, an $N$-dimensional phase space $\mathcal{M}$ and a reference orbit $\{x_n\}_{-\infty\leq n\leq \infty}$ passing through an arbitrary point $x_0$ should be considered.
The three conditions can be outlined as follows:
(1) The deviations from the reference orbit can be categorized into stable and unstable directions.
In the stable direction, an orbit passing through a point $y_0$ deviating from the reference orbit converges exponentially to the reference orbit in the forward direction of time.
Conversely, in the unstable direction, an orbit passing through a point $z_0$ deviating from the reference orbit 
diverges exponentially from the reference orbit in the forward direction of time.
(2) The expansion and contraction directions change continuously with respect to point $x$.
(3) Trajectories deviating in the stable direction continue to converge over time, whereas trajectories deviating in the unstable direction continue to diverge.

In the context of a mixing Anosov system, the equation governing the map $T:\mathcal{M}\to \mathcal{M}$ and the observable $F(\Gamma)$ is valid owing to the ergodic property \cite{gallavotti1999statistical}.
\begin{equation}
	\lim_{t\to \infty}\frac{1}{t}\sum_{n=0}^{t-1}F(T^n\Gamma) = \int_\mathcal{M} F(y)\mu(dy)
\end{equation}
The steady state represented by the probability distribution in this equation is known as the SRB distribution (SRB measure).
The SRB measure is characterized by being absolutely continuous along the unstable direction.
An SRB measure with no zero Lyapunov exponents is a physical measure \cite{pesin1976families,katok2006invariant,pugh1989ergodic,young2002srb}.
An ergodic invariant measure $\nu$ is deemed a physical measure \cite{young2002srb} if there exists a subset $V \subset \mathcal{M}$ with positive Lebesgue measure, such that 
for every continuous observable $\phi: \mathcal{M}\to \mathbb{R}$, $\lim_{n\to \infty}\frac{1}{n}\sum_{i=0}^{n-1}\phi(x_i)=\int_V \phi d\nu$.

The relaxation of the density function from a non-equilibrium to an equilibrium state through the attainment of an SRB distribution, facilitated by mixing, is elucidated from the perspective of chaotic dynamical systems. 

This study proposes a symplectic map derived from a Hamiltonian with time-reversal symmetry in which almost any initial density function is proven to converge to the uniform distribution.
Our findings reveal that this novel symplectic map exhibits time-reversal symmetry at the microscopic level while displaying irreversible behavior in terms of mixing properties at the macroscopic level. Therefore, we prove that under certain conditions, the novel symplectic map is an Anosov diffeomorphism, with the Lebesgue measure $\mu$ serving as the unique equilibrium state, unique SRB measure, and physical measure. 
Furthermore, we analytically demonstrate that the KS entropy is positive, and the critical exponent of the Lyapunov exponent is 1/2.

\section{Mechanics}
\label{sec:mechanics}
Let us consider the following Hamiltonian system such as
\begin{eqnarray}
H(p_1,p_{2}, q_1, q_{2})
&\equiv& \frac{p_1^2}{2}  + \frac{p_2^2}{2} +V(q_1, q_{2}), \label{Hamiltonian}\\
V(q_1, q_{2}) &\equiv& -\frac{\varepsilon}{\pi}\log\left|\cos\left\lbrace  \pi(q_1-q_2)\right\rbrace \right|,
\end{eqnarray}
where $\varepsilon \in \mathbb{R}$ represents a perturbation parameter and $p_1, p_2, q_1, q_2 \in \mathbb{R}$. The potential is periodic in nature.
The shapes of the periodic part of potential $V(q_1, q_2)$ at $\varepsilon = 2.0, -2.0$ are shown in Figures \ref{Fig: potential p} and \ref{Fig: potential m}, respectively.

\begin{figure}[h!]
		\begin{tabular}{cc}
			\begin{minipage}[t]{.45\columnwidth}
				\centering
				\includegraphics[width=\columnwidth]{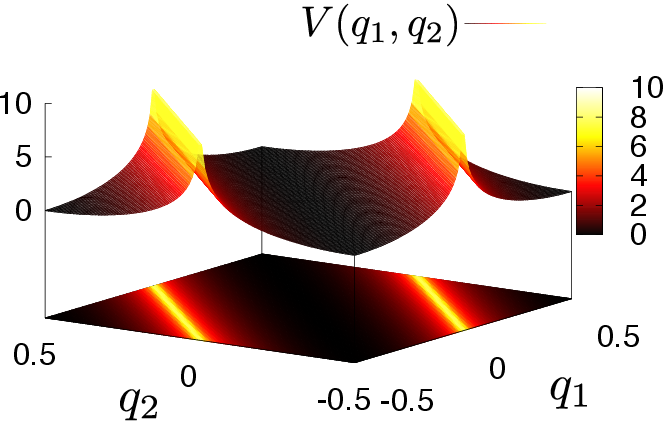}
				
				\vspace*{1cm}
				\caption{Shape of the periodic potential $V(q_1, q_2)$ for $\varepsilon=2.0$ (repulsive force occurs). }
				\label{Fig: potential p}    
			\end{minipage} &
			\hspace*{1cm}
			\begin{minipage}[t]{.45\columnwidth}
				\centering
				\includegraphics[width=\columnwidth]{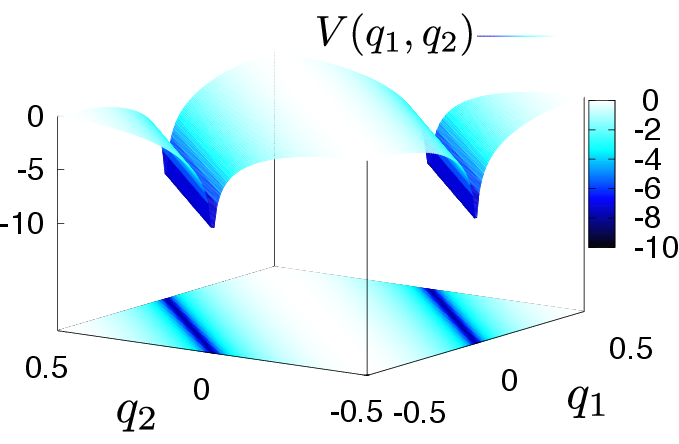}
				
				\vspace*{1cm}
				\caption{Shape of the {periodic} potential $V(q_1, q_2)$ for $\varepsilon=-2.0$ (attractive force occurs).}
				\label{Fig: potential m}    
			\end{minipage}
		\end{tabular}
\end{figure}
Next, we consider the following map $T_{\varepsilon, \Delta \tau}$ 
obtained {using} the leapfrog method (second-order {\it time-reversal symmetric} symplectic integrator)
for the canonical equations of motion associated with (\ref{Hamiltonian}).

\begin{definition}[$T_{\varepsilon, \Delta \tau}$]
The symplectic map 
$T_{\varepsilon, \Delta \tau} : \mathbb{R}^4\backslash \Gamma \to \mathbb{R}^4\backslash \Gamma$
is defined as follows:
{\footnotesize
\begin{equation}
\left(
\begin{array}{l}
p_1(n+1)\\
p_2(n+1)\\
q_1(n+1)\\
q_2(n+1)
\end{array}
\right)=
\left(
\begin{array}{l}
p_1(n) -\varepsilon \Delta \tau \tan\left[\pi\left\lbrace q_1(n)+ \frac{\Delta \tau}{2}p_1(n)-q_2(n)- \frac{\Delta \tau}{2}p_2(n)\right\rbrace \right] \\
p_2(n) +\varepsilon \Delta \tau \tan\left[\pi\left\lbrace q_1(n)+ \frac{\Delta \tau}{2}p_1(n)-q_2(n)- \frac{\Delta \tau}{2}p_2(n)\right\rbrace \right] \\
q_1(n) + p_1(n)\Delta \tau - \frac{\varepsilon(\Delta \tau)^2}{2} 
\tan\left[\pi\left\lbrace q_1(n)+ \frac{\Delta \tau}{2}p_1(n)-q_2(n)- \frac{\Delta \tau}{2}p_2(n)\right\rbrace \right]\\
q_2(n) + p_2(n)\Delta \tau + \frac{\varepsilon(\Delta \tau)^2}{2} 
\tan\left[\pi\left\lbrace q_1(n)+ \frac{\Delta \tau}{2}p_1(n)-q_2(n)- \frac{\Delta \tau}{2}p_2(n)\right\rbrace \right]
\end{array}
\right),	
\label{Symplectic map using Tan}
\end{equation}
}
where $\Delta \tau$ represents a step size and $\Gamma$ denotes a set $\{(p_1, p_2, q_1, q_2)\}$ in which
an integer $n$ exist, such that $T_{\varepsilon, \Delta\tau}^n(p_1, p_2, q_1, q_2) = (p_1', p_2', q_1', q_2')$, where
$q_1' + (\Delta\tau/2)p_1' -q_2' - (\Delta\tau/2)p_2' = (2m+1)/2, m \in \mathbb{Z}$.
\end{definition}

The Lebesgue measure of $\Gamma$ appears to be zero. 

To demonstrate the convergence of the initial density functions to the uniform distribution (SRB distribution) and positivity of the average amount of information per unit time in the dynamical system with time-reversal symmetry, we established that the dynamical system is an Anosov diffeomorphism \cite{katok1995introduction}.
To achieve this objective, alternative forms of the symplectic map, denoted as $\widehat{T}_{\varepsilon, \Delta\tau}$ and 
$\widetilde{T}_{\varepsilon, \Delta\tau}$ were deveoped from $T_{\varepsilon, \Delta\tau}$
to {streamline the proof process}. 
For the Hamiltonian $H$ and maps $T_{\varepsilon, \Delta\tau}$,
$\widehat{T}_{\varepsilon, \Delta\tau}$ and $\widetilde{T}_{\varepsilon, \Delta\tau}$,
the following diagram showcases their relationship:
\begin{displaymath}
\xymatrix{
	H\ar[r] &T_{\varepsilon, \Delta\tau} \ar[r] & \widehat{T}_{\varepsilon, \Delta \tau} \ar[r] &\widetilde{T}_{\varepsilon, \Delta \tau}
}
\end{displaymath}

Here, we introduce the map $\widehat{T}_{\varepsilon, \Delta\tau}$. Two relations exist within the map $T_{\varepsilon, \Delta\tau}$, such that
\begin{equation}
\begin{array}{lll}
p_1(n) + p_2(n) &=&  p_0 \equiv p_1(0)+ p_2(0),\\
q_1(n) + q_2(n)  &=& q_0 + np_0\Delta\tau, 
\end{array}
\end{equation}
where $q_0 \equiv q_1(0) +q_2(0)$.
Therefore, by {eliminating} $p_2$ and $q_2$ using {the} aforementioned equation, the  map can be defined as follows:
$\widehat{T}_{\varepsilon, \Delta\tau}$ as follows.

\begin{definition}[$\widehat{T}_{\varepsilon, \Delta \tau}$]
	{A} symplectic map $\widehat{T}_{\varepsilon, \Delta \tau}:M' \equiv \mathbb{R}^2 \backslash \Gamma'\to 
	M'$ {is defined as}
	{\footnotesize
		\begin{equation}
		\hspace*{-2.5cm}
		\begin{array}{ccl}
		\left(
		\begin{array}{r}
		p_1(n+1)\\
		q_1(n+1)
		\end{array}
		\right)
		&=&
		\left(
		\begin{array}{l}
		p_1(n) -\varepsilon (\Delta \tau)\tan\left[\pi\left\lbrace 2q_1(n)+ p_1(n)\Delta \tau -q_0 -p_0\Delta \tau\left(n + \frac{1}{2}\right)\right\rbrace \right]\\
		q_1(n) +p_1(n) \Delta\tau - 
		\frac{\varepsilon(\Delta \tau)^2}{2}\tan\left[\pi\left\lbrace 2q_1(n)+ p_1(n)\Delta \tau -q_0 -p_0\Delta \tau\left(n + \frac{1}{2}\right)\right\rbrace \right]
		\end{array}
		\right)
		\end{array}\label{T-hat honbun}
		\end{equation}
	}		
	\noindent where $\Gamma'$ represents a set $\{(p_1(n), q_1(n))\} \subset \mathbb{R}^2$
	such that an integer $k$ exists
	{that} satisfies the condition 
	\begin{equation}
	\widehat{T}_{\varepsilon, \Delta\tau}^k(p_1(n), q_1(n)) = (p_1(n+k), q_1(n+k)),
	\end{equation}
    where $2q_1(n+k) + p_1(n+k)\Delta\tau -q_0 -p_0\Delta\tau\left(n+k+\frac{1}{2}\right) = \frac{2m+1}{2}, ~~m\in \mathbb{Z}$.
\end{definition}
The map $T_{\varepsilon, \Delta \tau}$ can be reconstructed using $\widehat{T}_{\varepsilon, \Delta \tau}$
and the following relations: 
\begin{eqnarray}
\left\lbrace \begin{array}{lll}
p_1(n) + p_2(n) &=&  p_0 \equiv p_1(0)+ p_2(0),\\
q_1(n) + q_2(n) &=& q_0 + np_0\Delta\tau. 
\end{array}
\right.
\end{eqnarray}

Therefore, a new region is defined as $I_{\delta_{N,\Delta\tau,\varepsilon}}\equiv \left(-\frac{1}{2}+\frac{\delta_{N,\Delta\tau,\varepsilon}}{\pi},\frac{1}{2}-\frac{\delta_{N,\Delta\tau,\varepsilon}}{\pi}\right], 0<\delta_{N,\Delta\tau,\varepsilon}\ll1$, where $\delta_{N,\Delta\tau,\varepsilon}$ satisfies the following equations.
\begin{equation}
\begin{array}{rll}
F(\delta_{N,\Delta\tau,\varepsilon})\equiv\frac{\left|2\left(-1+\frac{2\delta_{N,\Delta\tau,\varepsilon}}{\pi}\right)+ 2\varepsilon(\Delta\tau)^2\left[\tan\left(\frac{\pi}{2}-\delta_{N,\Delta\tau,\varepsilon}\right)-\tan\left(-\frac{\pi}{2}+\delta_{N,\Delta\tau,\varepsilon}\right)\right]\right|}{\left(1-\frac{2\delta_{N,\Delta\tau,\varepsilon}}{\pi}\right)}
&=& N \in \mathbb{N}\\
\left|2\left(-1+\frac{2\delta_{N,\Delta\tau,\varepsilon}}{\pi}\right)+ 2\varepsilon(\Delta\tau)^2\left[\tan\left(\frac{\pi}{2}-\delta_{N,\Delta\tau,\varepsilon}\right)-\tan\left(-\frac{\pi}{2}+\delta_{N,\Delta\tau,\varepsilon}\right)\right]\right| &=& N \left(1-\frac{2\delta_{N,\Delta\tau,\varepsilon}}{\pi}\right).
\end{array}
\label{Eq: delta definition}
\end{equation}
The region $I_{\delta_{N,\Delta\tau,\varepsilon}}$ is a subset of region $I\equiv (-1/2,1/2]$.
As the map $F(\delta_{N,\Delta\tau,\varepsilon})$ diverges to positive infinity in the limit of $\delta_{N,\Delta\tau,\varepsilon}\to+0$, a value of $N$ exist that satisfies equation (\ref{Eq: delta definition}). If a considerably large value is assumed, the corresponding value of $\delta_{N,\Delta\tau,\varepsilon}$ that satisfies equation (\ref{Eq: delta definition}) decreases and region $I_{\delta_{N,\Delta\tau,\varepsilon}}$ can be approximated as $I$.
On $I_{\delta_{N,\Delta\tau,\varepsilon}}$, we identify
$-\frac{1}{2}+\frac{\delta_{N,\Delta\tau,\varepsilon}}{\pi}$ with $\frac{1}{2}-\frac{\delta_{N,\Delta\tau,\varepsilon}}{\pi}$ such that $I_{\delta_{N,\Delta\tau,\varepsilon}}$ is isomorphic to $\mathbb{S}^1$. The specific definition of the domain $I_{\delta_{N,\Delta\tau,\varepsilon}}$ permits the subsequent utilization of equation (\ref{Eq: delta definition}) to demonstrate that the map $\widetilde{T}_{\varepsilon, \Delta \tau}$ is a diffeomorphism on the domain $I_{\delta_{N,\Delta\tau,\varepsilon}}$.

The transformation $P$ that transfers the phase part of the tan function of the mapping $\widehat{T}_{\varepsilon, \Delta \tau}$ onto $I_{\delta_{N,\Delta\tau,\varepsilon}}$ is expressed as follows:

{\small
\begin{equation}
\left(
\begin{array}{c}
s_n'\\
t_n'
\end{array}
\right)=
P\left(
\begin{array}{c}
p_1(n)\\
q_1(n)
\end{array}
\right)
=
\left(
\begin{array}{c}
 2q_1(n)+ p_1(n)\Delta \tau -q_0 -p_0\Delta \tau \left(n+1/2\right)~\mbox{mod}~{I_{\delta_{N,\Delta\tau,\varepsilon}}}\\
Q(p_1(n), q_1(n)) -q_0 -p_0\Delta \tau \left(n + 3/2\right) \bmod {I_{\delta_{N,\Delta\tau,\varepsilon}}}
\end{array}
\right), \label{homeomorphism}
\end{equation}
}
\noindent where $Q(p_1(n), q_1(n)) \equiv 2q_1(n+1)+ p_1(n+1)\Delta \tau$.
`` $\mbox{mod}~I_{\delta_{N,\Delta\tau,\varepsilon}}$'' is defined as follows.
\begin{definition}[mod $I_{\delta_{N,\Delta\tau,\varepsilon}}$]
		We define the operation $\bmod~I_{\delta_{N,\Delta\tau,\varepsilon}} : \mathbb{R} \to I_{\delta_{N,\Delta\tau,\varepsilon}}$ as
		\begin{equation}
		\begin{array}{rll}
		x ~\bmod~I_{\delta_{N,\Delta\tau,\varepsilon}} &\equiv& 
		x - n\left(1-\frac{2}{\pi}\delta_{N,\Delta\tau,\varepsilon}\right), \\
		~\mbox{where}~
		-\frac{1}{2}+\frac{\delta_{N,\Delta\tau,\varepsilon}}{\pi}+ n\left(1-\frac{2}{\pi}\delta_{N,\Delta\tau,\varepsilon}\right) &<& x \leq \frac{1}{2}-\frac{\delta_{N,\Delta\tau,\varepsilon}}{\pi} + n\left(1-\frac{2}{\pi}\delta_{N,\Delta\tau,\varepsilon}\right),
		n \in \mathbb{Z}\nonumber.
		\end{array}
		\end{equation}
\end{definition}

Regarding the time evolution of $(s_n',t_n')$ previously defined, we define a new symplectic map $\widetilde{T}_{\varepsilon, \Delta \tau}:  M \equiv I_{\delta_{N,\Delta\tau,\varepsilon}}^2 \to  M$
obtained from $\widehat{T}_{\varepsilon, \Delta \tau}$ as follows:

\begin{definition}[$\widetilde{T}_{\varepsilon, \Delta \tau}$]
	{We define} a symplectic map $\widetilde{T}_{\varepsilon, \Delta\tau}:M \to M$ {as} follows:
	\begin{equation}
	\hspace*{-2cm}
	\left(
	\begin{array}{c}
	s_{n+1},\\
	t_{n+1}
	\end{array}
	\right)=
	\left(
	\begin{array}{c}
	g(t_n)\\
	h(s_n, t_n)
	\end{array}
	\right)
	=
	\left(
	\begin{array}{c}
	t_n\\
	2t_n -s_n -2\varepsilon(\Delta \tau)^2\tan(\pi t_n) \bmod {I_{\delta_{N,\Delta\tau,\varepsilon}}}
	\end{array}
	\right) \label{2D-map}
	\end{equation}
	where $x_n \equiv (s_n, t_n) \in I_{\delta_{N,\Delta\tau,\varepsilon}}^2$.
\end{definition}

The symplectic map $\widetilde{T}_{\varepsilon, \Delta \tau}$ is topologically semi-conjugate to
$\widehat{T}_{\varepsilon, \Delta \tau}$
through a transformation denoted as $P$. The $(p_1,q_1)$ utilized at the transformation $P$ is not defined on $\Gamma'$; therefore, the point $(s',t')$ corresponding to a point on $\Gamma'$ is not defined. However, the mapping $\widetilde{T}_{\varepsilon, \Delta \tau}$ is not directly defined from the time evolution of $(s',t')$ but using equation (\ref{2D-map}).

For the symplectic maps $T_{\varepsilon, \Delta\tau}$, $\widehat{T}_{\varepsilon, \Delta\tau}$, and
$\widetilde{T}_{\varepsilon, \Delta\tau}$  the following diagram is satisfoed:
\begin{eqnarray*}
	\begin{CD}
		\left(\mathbb{R}^2\times I_{\delta_{N,\Delta\tau,\varepsilon}}^2\right)\backslash \Gamma @>T_{\varepsilon, \Delta\tau} >> 
		\left(\mathbb{R}^2\times I_{\delta_{N,\Delta\tau,\varepsilon}}^2\right)\backslash \Gamma
	\end{CD}
\end{eqnarray*}
\begin{eqnarray*}
	\begin{CD}
		\left(\mathbb{R}\times I_{\delta_{N,\Delta\tau,\varepsilon}}\right)\backslash \Gamma' @>\widehat{T}_{\varepsilon, \Delta\tau} >> 
		\left(\mathbb{R}\times I_{\delta_{N,\Delta\tau,\varepsilon}}\right) \backslash \Gamma'\\
		@VPVV   @VVPV \\
		I_{\delta_{N,\Delta\tau,\varepsilon}}^2 @>\widetilde{T}_{\varepsilon, \Delta\tau}>> I_{\delta_{N,\Delta\tau,\varepsilon}}^2
	\end{CD}
\end{eqnarray*}

\noindent {The} Jacobian of the map $\widetilde{T}_{\varepsilon, \Delta \tau}$ is expressed as follows:
\begin{eqnarray}
J(x_n) = \left(
\begin{array}{cc}
0 &  1\\
-1 & 2-\frac{2\pi \varepsilon(\Delta \tau)^2}{\cos^2(\pi t_n)}
\end{array}
\right), \label{Jacobi}
\end{eqnarray}
and it is not a constant matrix. 
In the case of constant matrices, Arnold's cat map \cite{arnold1968ergodic,ermann2012arnold,barash2006periodic} and (multi) baker's map \cite{dorfman1999introduction,tasaki1998analytical,tasaki1995fick} have been extensively investigated.
The linear instability condition in which at least one eigenvalue of the Jacobian has an absolute value greater than unity, for the map $\widetilde{T}_{\varepsilon, \Delta \tau}$ is
obtained as follows:
\begin{eqnarray}
\varepsilon<0~~\mbox{or}~~ \frac{2}{\pi(\Delta \tau)^2} < \varepsilon, \label{Condition}
\end{eqnarray}
This condition is associated with the map being an Anosov diffeomorphism (Anosov system).

{Because} the map $\widetilde{T}_{\varepsilon, \Delta \tau}$ is symplectic, 
$\widetilde{T}_{\varepsilon, \Delta \tau}$ preserves the Lebesgue measure $\mu$ on $I_{\delta_{N,\Delta\tau,\varepsilon}}^2$.

\section{Time-reversal symmetry}
\label{sec: time-reversal symmetry}

In this section, we demonstrated that the symplectic map $\widetilde{T}_{\varepsilon,\Delta\tau}$ utilized in this study exhibits time-reversal symmetry.
\begin{definition}\label{Def: time-reversal symmetry}
    Let $f: M\to M$ be a map. The map $f$ exhibits time-reversal symmetry when a map $R$ \cite{roberts1992chaos} exists such that
    \begin{equation}
        R\circ f = f^{-1}\circ R, R\circ R = Id.
    \end{equation}
\end{definition}
When map $R$ is present, the form of $f^{-1}\circ R$ is equivalent to that of $R\circ f$. Therefore, any orbit obtained by utilizing $\{f^{-n}\}_{n\in \mathbb{Z}}$ can also be  obtained by utilizing $\{R \circ f^n \circ R^{-1}\}_{n\in \mathbb{Z}}$. We validate the time-reversal symmetry property of the map $\widetilde{T}_{\varepsilon,\Delta\tau}$  for $R_1(s,t) = (-t,-s)$.

For any point $(s,t)$, 
\begin{equation}
R_1 \circ \widetilde{T}_{\varepsilon,\Delta\tau} \left(
    \begin{array}{c}
         s  \\
         t 
    \end{array}
    \right)= \left(
    \begin{array}{c}
         -2t +s +2\varepsilon(\Delta\tau)^2\tan(\pi t) \bmod I_{\delta_N,\Delta \tau}  \\
         -t 
    \end{array}
    \right). \label{Eq: T tilde time reversal symmetry1}
\end{equation}
The following equation can be obtained:
\begin{equation}
\widetilde{T}_{\varepsilon,\Delta\tau}^{-1}\circ R_1 \left(
\begin{array}{c}
     s  \\
     t 
\end{array}
\right) = \left(
    \begin{array}{c}
         -2t+s+2\varepsilon(\Delta\tau)^2\tan(\pi t) \bmod I_{\delta_N,\Delta \tau} \\
         -t
    \end{array}
\right).\label{Eq: T tilde time reversal symmetry2}
\end{equation}
Therefore, from \ref{Def: time-reversal symmetry}, Eq. (\ref{Eq: T tilde time reversal symmetry1}), and (\ref{Eq: T tilde time reversal symmetry2}), the 
map $\widetilde{T}_{\varepsilon,\Delta\tau}$ exhibits time-reversal symmetry.

\section{Compactness and diffeomorphism}
\label{sec: compactness and diffeomorphism}
Here, we establish that the domain $M$ is compact and that $\widetilde{T}_{\varepsilon, \Delta\tau}$ is a diffeomorphism on $M$.

The domain $I_{\delta_{N,\Delta\tau,\varepsilon}}$ is isomorphic to the circle $\mathbb
{S}^1$ owing to the identification of $-\frac{1}{2}+ \frac{\delta_{N,\Delta\tau,\varepsilon}}{\pi}$ with $\frac{1}{2}-\frac{\delta_{N,\Delta\tau,\varepsilon}}{\pi}$.
Therefore, the domain $M$ is topologically conjugate with 
torus $\mathbb{T}^2$, with $M$ being compact.

Next, we demonstrate that the map $\widetilde{T}_{\varepsilon, \Delta \tau}$ is a diffeomorphism on $M$.
To establish this, we must first establish that the map $\widetilde{T}_{\varepsilon, \Delta\tau}$ is 
surjective, injective, and smooth, and that its inverse is continuous.

\noindent	(Proof of surjectivity):
For any point $(s_{n+1}, t_{n+1}) \in M$, a point, $(s_n, t_n) \in M$ exists such that
\begin{equation}
\left\lbrace 
\begin{array}{ccl}
s_{n} &=& 2s_{n+1} - 2\varepsilon(\Delta \tau)^2\tan(\pi s_{n+1}) -t_{n+1} ~\mbox{mod}~{I_{\delta_{N,\Delta\tau,\varepsilon}}},\\
t_{n} &=& s_{n+1}.
\end{array}
\right.
\end{equation}
Therefore, the map $\widetilde{T}_{\varepsilon, \Delta \tau}$ is surjective.

\noindent	(Proof of injectivity):
The injectivity of the map $\widetilde{T}_{\varepsilon, \Delta\tau}$ can be established through a proof by contradiction.
Assume $(s_n, t_n)$ and $(s_n', t_n')$ are two points {that} cannot be identified on $M$.
Assuming the image $(s_{n+1}, t_{n+1})$ is identified with $(s_{n+1}', t_{n+1}')$, then
\begin{eqnarray}
s_{n+1} &=& s_{n+1}',\nonumber\\
\mbox{Thus},~~ t_n &=& t_n'. \label{injective}
\end{eqnarray}
Because $2t-2\varepsilon(\Delta\tau)^2\tan(\pi t)$ is finite for $t \in I_{\delta_{N,\Delta\tau,\varepsilon}}$ and (\ref{injective}), the relationship $s_n = s_n'$ can be obtained. These results contradict the condition $(s_n, t_n) \neq (s_n', t_n')$.
Therefore, the map $\widetilde{T}_{\varepsilon, \Delta \tau}$ is injective.

\noindent	(Proof of smoothness):
Evidently,  $g, h, \frac{\partial g}{\partial t}$, $\frac{\partial h}{\partial s}$, and $\frac{\partial h}{\partial t}$ defined in (\ref{2D-map}) 
are continuous on $M$ except for $(s, t)$, where $s = \frac{1}{2}-\frac{\delta_{N,\Delta\tau,\varepsilon}}{\pi}$ or $t =\frac{1}{2}-\frac{\delta_{N,\Delta\tau,\varepsilon}}{\pi}$.
Regarding the smoothness at $\frac{1}{2}-\frac{\delta_{N,\Delta\tau,\varepsilon}}{\pi}$, $g(t)=t$ and $\frac{\partial g}{\partial t}$ are continuous at $t = \frac{1}{2}-\frac{\delta_{N,\Delta\tau,\varepsilon}}{\pi}$
and similarly, $h(s, t)=2t - 2\varepsilon(\Delta \tau)^2\tan(\pi t)-s$ and $\frac{\partial h}{\partial s}$ are 
continuous at $s=\frac{1}{2}-\frac{\delta_{N,\Delta\tau,\varepsilon}}{\pi}$ 
Next, the continuity of $h(s, t)$ and $\frac{\partial h}{\partial t}$ at $t = \frac{1}{2} - \frac{\delta_{N,\Delta\tau,\varepsilon}}{\pi}$ is examined.
From the definition of $\delta_{N,\Delta\tau,\varepsilon}$, the following equation can be obtained:
\begin{equation}
	\begin{array}{ccl}
		&&\left|h\left(s, -\frac{1}{2}+\frac{\delta_{N,\Delta\tau,\varepsilon}}{\pi}\right)-h\left(s, \frac{1}{2}-\frac{\delta_{N,\Delta\tau,\varepsilon}}{\pi}\right)\bmod ~{I_{\delta_{N,\Delta\tau,\varepsilon}}}\right|\\
		&=& \left|\left[2\left(-1+\frac{2\delta_{N,\Delta\tau,\varepsilon}}{\pi}\right)+ 2\varepsilon(\Delta\tau)^2\left\lbrace\tan\left(\frac{\pi}{2}-\delta_{N,\Delta\tau,\varepsilon}\right)-\tan\left(-\frac{\pi}{2}+\delta_{N,\Delta\tau,\varepsilon}\right)\right\rbrace \right] \bmod I_{\delta_{N,\Delta\tau,\varepsilon}}\right|\\
        &=& \left|N \left(1-\frac{2\delta_{N,\Delta\tau,\varepsilon}}{\pi}\right) \bmod I_{\delta_{N,\Delta\tau,\varepsilon}}\right|~(\because (\ref{Eq: delta definition}))\\
        &=& 0.
	\end{array} \nonumber
\end{equation}
Therefore, $h(s, t)$ is continuous at $t=\frac{1}{2} - \frac{\delta_{N,\Delta\tau,\varepsilon}}{\pi}$ on torus $M$.
{Because} $h(s, t)$ is an odd function with $t$, the following can be obtained:
$\frac{\partial h}{\partial t}\left(s, -\frac{1}{2}+\frac{\delta_{N,\Delta\tau,\varepsilon}}{\pi}\right)=\frac{\partial h}{\partial t}\left(s, \frac{1}{2}-\frac{\delta_{N,\Delta\tau,\varepsilon}}{\pi}\right)$.
Hence, $\frac{\partial h}{\partial t}$ is continuous at $t=\frac{1}{2}-\frac{\delta_{N,\Delta\tau,\varepsilon}}{\pi}$.
The aforementioned properties also apply to the inverse map $\widetilde{T}_{\varepsilon, \Delta \tau}^{-1}$. Therefore, the map $\widetilde{T}_{\varepsilon, \Delta \tau}$ is a $C^1$ diffeomorphism on $M$.
Next, we indicate that the map $\widetilde{T}_{\varepsilon, \Delta \tau}$ is a $C^{1+1}$ diffeomorphism on $M$. Obviously, $g$ is a $C^2$ diffeomorphism in terms of $s,t$ and $h$ is a $C^2$ diffeomorphism in terms of $s$. The function $\frac{\partial^2 h}{\partial t^2}$ is continuous except for $t=\frac{1}{2}-\frac{\delta_{N,\Delta\tau,\varepsilon}}{\pi}$. The function $\frac{\partial h}{\partial t}$ is bounded on $M$. Therefore, positive constants $K_1$ and $K_2$ exist such that
\begin{equation}
    \begin{array}{lll}
    \displaystyle\lim_{t \to \frac{1}{2}-\frac{\delta_{N,\Delta\tau,\varepsilon}}{\pi}-0} \frac{\left|\frac{\partial h}{\partial t}\left(s,\frac{1}{2}-\frac{\delta_{N,\Delta\tau,\varepsilon}}{\pi}\right)-\frac{\partial h}{\partial t}(s,t)\right|}{\left|\frac{1}{2}-\frac{\delta_{N,\Delta\tau,\varepsilon}}{\pi}-t\right|} &\leq& K_1,\\
    \displaystyle\lim_{t \to -\frac{1}{2}+\frac{\delta_{N,\Delta\tau,\varepsilon}}{\pi}+0} \frac{\left|\frac{\partial h}{\partial t}\left(s,-\frac{1}{2}+\frac{\delta_{N,\Delta\tau,\varepsilon}}{\pi}\right)-\frac{\partial h}{\partial t}(s,t)\right|}{\left|-\frac{1}{2}+\frac{\delta_{N,\Delta\tau,\varepsilon}}{\pi}-t\right|} &\leq& K_2.
    \end{array} \nonumber
\end{equation}
Therefore, a positive constant $C$ exists, such that
\begin{equation}
    \frac{\left|\frac{\partial h}{\partial t}(s,x)-\frac{\partial h}{\partial t}(s,y)\right|}{|x-y|}\leq C. \nonumber
\end{equation}
The aforementioned properties also apply to the inverse map $\widetilde{T}_{\varepsilon, \Delta \tau}^{-1}$.
Therefore, the map $\widetilde{T}_{\varepsilon, \Delta \tau}$ is a $C^{1+1}$ diffeomorphism, implying that the map $\widetilde{T}_{\varepsilon, \Delta \tau}$ is a $C^{1+\alpha}$ diffeomorphism.

\section{Proof of Anosov diffeomorphism}
\label{sec: Proof of Anosov diffeomorphism}

In this section, we aim to establish that the map $\widetilde{T}_{\varepsilon, \Delta \tau}$ is an Anosov diffeomorphism under the condition of linear instability. This is crucial in demonstrating the irreversible behavior of the dynamical system with time-reversal symmetry  from a macroscopic perspective. While Gallavotti postulated the \textit{chaotic hypothesis} \cite{gallavotti2014nonequilibrium} for many-particle systems to possess Anosov properties, we provide a direct proof of the Anosov properties for a specific Hamiltonian system with limited degrees of freedom through the following theorems.

\setcounter{theorem}{0}
\renewcommand{\thetheorem}{\Alph{theorem}}
\begin{theorem}
	Assuming the condition $\varepsilon< 0~\mbox{or}~ \frac{2}{\pi(\Delta \tau)^2}<\varepsilon$ is satisfied, the map $\widetilde{T}_{\varepsilon, \Delta \tau}$ on $M$ qualifies as an Anosov diffeomorphism.
\end{theorem}
\setcounter{theorem}{5}
\renewcommand{\thetheorem}{\arabic{theorem}}

Before delving into the proof, we note the following: a point, $(s_n, t_n)$, of $M$ is denoted as $x_n = (s_n, t_n)$.
The Jacobian of the map $\widetilde{T}_{\varepsilon, \Delta \tau}$ is expressed as follows:
\begin{eqnarray}
	J(x_n) = \left(
	\begin{array}{cc}
		0 & 1\\
		-1 & 2- \frac{2\pi \varepsilon(\Delta \tau)^2}{\cos^2(\pi t_n)}
	\end{array}
	\right), \label{Jacobi2}
\end{eqnarray}
where the components of the Jacobian $J(x_n)$ are on $\mathbb{R}$. The linear instability condition, where at least one eigenvalue of the matrix (\ref{Jacobi2}), has an  absolute value greater than unity
for the map $\widetilde{T}_{\varepsilon, \Delta \tau}$ is expressed as follows:
\begin{eqnarray}
	\varepsilon<0 ~~\mbox{or}~~\frac{2}{\pi(\Delta \tau)^2}< \varepsilon. \label{Condition2}
\end{eqnarray}
The linear instability condition is related to the condition that the map 
$\widetilde{T}_{\varepsilon, \Delta\tau}$ is an Anosov diffeomorphism.

An Anosov diffeomorphism is defined as \cite{katok1995introduction,gallavotti2014nonequilibrium,franks1969anosov,bowen1975equilibrium}, with the  theorem proven as follows:

\begin{definition}[Anosov diffeomorphism]\label{Def: Anosov diffeomorphism}
	Let $\mathcal{M}^*$ is a compact manifold, $\|\cdot\|$ an Euclidean norm, and $f: \mathcal{M}^*\to \mathcal{M}^*$ a diffeomorphism. 
	At each point $x\in \mathcal{M}^*$, consider the tangent space, $T_x\mathcal{M}^*$.
	{$f$ is considered an Anosov diffeomorphism if it} satisfies the following conditions:
	\begin{enumerate}
		\renewcommand{\labelenumi}{(\Roman{enumi})}
		\item Subspaces $E_x^u$ and $E_x^s$ exist such that $T_x\mathcal{M}^* = E_x^u \oplus E_x^s$,
		\item $(D_xf)E_x^u = E_{f(x)}^u$,
		$(D_xf)E_x^s = E_{f(x)}^s$,
		\item $K>0$ and $0<\eta<1$ exist determined by $x$ and not by $\bm{\xi}$ or $n$ such that
		\begin{eqnarray}
			\bm{\xi}\in E_x^u \Longrightarrow \|(Df^n)\bm{\xi}\| &\geq& K \left(\frac{1}{\eta}\right)^n \|\bm{\xi}\|, \label{stretching}\\ 
			\bm{\xi}\in E_x^s \Longrightarrow \|(Df^{-n})\bm{\xi}\| &\geq&  K \left(\frac{1}{\eta}\right)^n \|\bm{\xi}\|. \label{shrinking}
		\end{eqnarray}
	\end{enumerate}
\end{definition}

The proof that the map $\widetilde{T}_{\varepsilon, \Delta \tau}$ is an Anosov diffeomorphism for $\varepsilon< 0, \frac{2}{\pi(\Delta \tau)^2}<\varepsilon$ is based on the concept presented in \cite{zaslavsky1985chaos}. 
While the Jacobians for Arnold's cat map \cite{arnold1968ergodic,ermann2012arnold} and multi-baker's map \cite{tasaki1995fick} remain constant, the Jacobians for our map vary at different points, 
and the aforementioned maps are extremely ideal compared with 
our maps. Compared with Arnold's cat and multi-baker's maps, our maps are more complex and challenging to establish as Anosov diffeomorphisms.
To establish this, the conditions \textit{(I)}, \textit{(II)}, and \textit{(III)} in Definition \ref{Def: Anosov diffeomorphism} should be satisfied. 
To validate the fulfilment of these conditions, we first define the cones $L^+$ and $L^-$, followed by the presentation of the following Lemmas:

\begin{definition}
	For every two-dimensional nonzero vector on $T_{x_n}M, x_n \in M, n \in \mathbb{Z}$,
	\begin{eqnarray*}
		\bm{a}_n= (a_1(n), a_2(n))^T \neq \bm{O},
	\end{eqnarray*}
	cones $L^+$ and $L^-$ are defined as follows:
	\begin{eqnarray*}
		L^+(x_n)  &=& \{(a_1(n), a_2(n))^T: \| J(x_n)\bm{a}_n\| > \|\bm{a}_n\|\},\\
		L^-(x_n)  &=& \{(a_1(n), a_2(n))^T: \|J(\widetilde{T}_{\varepsilon,\Delta \tau}^{-1}x_n)^{-1}\bm{a}_n\| > \|\bm{a}_n\|\},
	\end{eqnarray*}
	where $\|\cdot\|$ represents $\|\bm{a}(n)\|\equiv \sqrt{a_1^2(n)+a_2^2(n)}$.
\end{definition}

\begin{lemma}\label{Lemma:L+}
	The satisfaction of the condition $\varepsilon< 0, \frac{2}{\pi(\Delta \tau)^2}<\varepsilon$ results in the following equation: 
	\begin{eqnarray}
		{}^\forall\bm{a}_n \in L^+(x_n), J(x_n)\bm{a}_n \in L^+(\widetilde{T}_{\varepsilon,\Delta \tau} x_n)=L^+(x_{n+1}),  \label{Mixing condition1}
	\end{eqnarray}
\end{lemma}
\noindent This Lemma corresponds to condition \textit{(II)}.

\begin{proof}
	From simple calculations, $\bm{a}_n \notin L^+(x_n)$ and $\bm{a}_n \notin L^-(x_n)$ 
	when $a_2(n) =0$ and $a_1(n) = 0$.
	Suppose $a_2(n) \neq 0$ and $a_1(n) \neq 0$. In that case, the condition $\bm{a}_n \in L^+(x_n)$ in (\ref{Mixing condition1}) is expressed as follows:
	\begin{equation}
		\bm{a}_n \in L^+(x_n)
		\Leftrightarrow \left\lbrace 
		\begin{array}{cl}
			\frac{a_1(n)}{a_2(n)} > 1-\frac{\pi\varepsilon(\Delta \tau)^2}{\cos^2(\pi t_n)},& ~\mbox{for}~\varepsilon > \frac{2}{\pi(\Delta \tau)^2},\\
			\frac{a_1(n)}{a_2(n)} < 1-\frac{\pi\varepsilon(\Delta \tau)^2}{\cos^2(\pi t_n)},& ~\mbox{for}~\varepsilon < 0.
		\end{array}
		\right. \label{L^+}
	\end{equation}
	Selecting $\bm{a}_{n+1}$ as $J(x_n)\bm{a}_n = (a_1(n+1), a_2(n+1))$, the condition can be expressed as follows: 
	$\bm{a}_{n+1} \in L^+(\widetilde{T}_{\varepsilon, \Delta \tau} x_n)$ in (\ref{Mixing condition1}) as
	\begin{equation}
		\hspace*{-2cm}
		\bm{a}_{n+1} \in L^+(\widetilde{T}_{\varepsilon, \Delta \tau} x_n)=L^+(x_{n+1})
		\Leftrightarrow 
		\left\lbrace
		\begin{array}{lll}
			\frac{a_1(n+1)}{a_2(n+1)} &> 1-\frac{\pi\varepsilon(\Delta \tau)^2}{\cos^2(\pi t_{n+1})},&~\mbox{for}~\varepsilon > \frac{2}{\pi(\Delta \tau)^2},\\
			\frac{a_1(n+1)}{a_2(n+1)} &< 1-\frac{\pi\varepsilon(\Delta \tau)^2}{\cos^2(\pi t_{n+1})},&~\mbox{for}~\varepsilon < 0.
		\end{array}
		\right. 
		\label{L^+2}
	\end{equation}
	To establish (\ref{Mixing condition1}), we demonstrate that (\ref{L^+}) implies (\ref{L^+2}).
	
	\noindent 	by substituting $a_1(n+1)=a_2(n)$ and $a_2(n+1)= -a_1(n) +2\left(1-\frac{\pi\varepsilon(\Delta \tau)^2}{\cos^2(\pi t_n)}\right)a_2(n)$
	into (\ref{L^+2}), we obtain
	\begin{equation}
		\hspace*{-1cm}
		\frac{a_1(n+1)}{a_2(n+1)}=\frac{a_2(n)}{-a_1(n) +2\left(1-\frac{\pi\varepsilon(\Delta \tau)^2}{\cos^2(\pi t_n)}\right)a_2(n)}
		=\frac{1}{-\frac{a_1(n)}{a_2(n)} +2\left(1-\frac{\pi\varepsilon(\Delta \tau)^2}{\cos^2(\pi t_n)}\right)}.
	\end{equation}
	\noindent Therefore, when the condition in Eq. (\ref{L^+}) is satisfied, the condition in Eq. (\ref{L^+2}) is also satisfied, as it can be shown inductively that 
	\begin{eqnarray*}
		\left\lbrace 
		\begin{array}{lll}
			~{}^\forall k\geq n+1, & 1-\frac{\pi\varepsilon(\Delta \tau)^2}{\cos^2(\pi t_{k})}<\frac{a_1(k)}{a_2(k)}<0,&~\mbox{for}~\varepsilon>\frac{2}{\pi(\Delta \tau)^2},\\
			~{}^\forall k\geq n+1, & 0<\frac{a_1(k)}{a_2(k)}<1-\frac{\pi\varepsilon(\Delta \tau)^2}{\cos^2(\pi t_{k})},&~\mbox{for}~\varepsilon<0.
		\end{array}
		\right.
	\end{eqnarray*}
	Therefore, the condition in Eq. (\ref{Mixing condition1}) is satisfied. 	
\end{proof}

\begin{lemma}\label{Lemma:L-}
	When the condition $\varepsilon< 0, \frac{2}{\pi(\Delta \tau)^2}<\varepsilon$ is satisfied, the following equation is obtained: 
	\begin{eqnarray}
		{}^\forall \bm{a}_n \in L^-(x_n), J(x_{n-1})^{-1}\bm{a}_n =\bm{a}_{n-1} \in L^-(\widetilde{T}_{\varepsilon,\Delta \tau}^{-1} x_n)=L^-(x_{n-1}) . \label{Mixing condition2}
	\end{eqnarray}
\end{lemma}

\begin{proof}
	The conditions $\bm{a}_n \in L^-(x_n)$ and 
	$J(x_{n-1})^{-1}\bm{a}_n \in L^-(\widetilde{T}_{\varepsilon, \Delta \tau}^{-1}x_n)$ in (\ref{Mixing condition2}) are expressed as follows:
	\begin{equation}
		\bm{a}_n \in L^-(x_n) \Leftrightarrow
		\left\lbrace 
		\begin{array}{cl}
			\frac{a_2(n)}{a_1(n)} > 1-\frac{\pi\varepsilon(\Delta \tau)^2}{\cos^2(\pi t_{n-1})},& ~\mbox{for}~ \varepsilon > \frac{2}{\pi(\Delta \tau)^2},\\
			\frac{a_2(n)}{a_1(n)} < 1-\frac{\pi\varepsilon(\Delta \tau)^2}{\cos^2(\pi t_{n-1})},& ~\mbox{for}~ \varepsilon < 0,
		\end{array}
		\right. \label{L^-}
	\end{equation}
	and 
	\begin{equation}
		\begin{array}{lll}
			&&\bm{a}_{n-1}\equiv J(x_{n-1})^{-1}\bm{a}_n \in L^-(\widetilde{T}_{\varepsilon, \Delta \tau}^{-1}x_n)=L^-(x_{n-1})\\
			&\Leftrightarrow & \left\lbrace 
			\begin{array}{cl}
				\frac{a_2(n-1)}{a_1(n-1)} > 1-\frac{\pi\varepsilon(\Delta \tau)^2}{\cos^2(\pi t_{n-2})},& ~\mbox{for}~ \varepsilon > \frac{2}{\pi(\Delta \tau)^2},\\
				\frac{a_2(n-1)}{a_1(n-1)} < 1-\frac{\pi\varepsilon(\Delta \tau)^2}{\cos^2(\pi t_{n-2})},& ~\mbox{for}~ \varepsilon < 0.
			\end{array}
			\right.
		\end{array}
		\label{L^-2}
	\end{equation}
	
	By substituting $a_1(n-1) = 2\left(1-\frac{\pi\varepsilon(\Delta \tau)^2}{\cos^2(\pi t_{n-1})}\right)a_1(n)-a_2(n)$ and 
	$a_2(n-1) = a_1(n)$ into (\ref{L^-2}), we can deduce that the relationship in Eq. (\ref{L^-2}) is valid when the relationship in Eq. (\ref{L^-}) is satisfied because by induction:
    \begin{eqnarray*}
        \left\lbrace
        \begin{array}{lll}
             {}^\forall k \leq n-1, &1-\frac{\pi\varepsilon(\Delta \tau)^2}{\cos^2(\pi t_{k-1})} < \frac{a_2(k)}{a_1(k)}<0,&~\mbox{for}~ \varepsilon > \frac{2}{\pi(\Delta \tau)^2},\\
             {}^\forall k \leq n-1, &0< \frac{a_2(k)}{a_1(k)}<1-\frac{\pi\varepsilon(\Delta \tau)^2}{\cos^2(\pi t_{k-1})},&~\mbox{for}~ \varepsilon<0.
        \end{array}
        \right. 
    \end{eqnarray*}
\end{proof}

\begin{definition}
A subset $LL_{1}^{-}(x_n)$ of $L^-(x_n)$ is defined as follows:
	\begin{equation}
			LL_{1}^-(x_n)
			\equiv
			\left\lbrace 
			\begin{array}{ll}
				\left\lbrace (a_1(n), a_2(n)) :  -1<\frac{a_2(n)}{a_1(n)}<0\right\rbrace  & 
				\mbox{for}~\varepsilon> \frac{2}{\pi(\Delta \tau)^2},\\
				\left\lbrace (a_1(n), a_2(n)) :  0<\frac{a_2(n)}{a_1(n)}< 1\right\rbrace  & 
				\mbox{for}~\varepsilon<0.
			\end{array}
			\right. 
         \label{LL-1}
	\end{equation}
\end{definition}

After a straightforward calculation, we obtain
\begin{eqnarray}
	{}^\forall\bm{a}_n \in LL_{1}^-(x_n), J(x_{n-1})^{-1}\bm{a}_n \in LL_{1}^-(x_{n-1}). \label{LL-1-1 honbun}
\end{eqnarray}

\begin{lemma}\label{Lemma:LL-}
	Assuming the condition $\varepsilon< 0, \frac{2}{\pi(\Delta \tau)^2}<\varepsilon$ is satisfied, for any point $x_n$ and any integer $k \in \mathbb{Z}\geq n$, a deviation vector $\bm{a}_n$ exists such that
	\begin{eqnarray}
		\displaystyle \prod_{i=k}^{n}J(x_i)\bm{a}_n \in LL_1^-(x_{k+1}),~\bm{a}_n \in LL_1^-(x_n).
	\end{eqnarray}
\end{lemma}

\begin{proof}
	We employed a proof by contradiction.
	Assuming
	\begin{equation}
		{}^\exists y_n, {}^\exists k' \in \mathbb{Z}\geq n, {}^\forall \bm{a}_n \in LL_1^-(y_n)
		~~\mbox{s.t.}~ \prod_{i=k'}^{n}J(y_i)\bm{a}_n \notin LL_1^-(y_{k'+1}). \label{Assumption}
	\end{equation}
	An orbit, $\{y_m\}$, can be obtained, including $y_n$ and $y_{n'}$, where $n' > k'+1 > n$. Owing to the existence of a cone at all points, 
	$LL_1^-${. Therefore,} 
	a vector $\bm{A}_{n'} \in LL_1^-(y_{n'})$ exists.  Therefore, $\bm{A}_{k'+1}$ and $\bm{A}_{n}$ can be selected as follows:
	\begin{eqnarray}
		\bm{A}_{k'+1} &\equiv& \prod_{i=k'+1}^{n'-1}J(y_i)^{-1}\bm{A}_{n'} \in LL_{1}^-(y_{k'+1})~\because (\ref{Mixing condition2}),\nonumber\\
		\bm{A}_{n} &\equiv& \prod_{i=n}^{n'-1}J(y_i)^{-1}\bm{A}_{n'} \in LL_{1}^-(y_n),\nonumber\\
		\mbox{Therefore},~\prod_{i=k'}^{n}J(y_i)\bm{A}_n &\in& LL_{1}^-(y_{k'+1}). \label{contradiction}
	\end{eqnarray}
	Equation (\ref{contradiction}) is contradictory to (\ref{Assumption}). Therefore, Lemma \ref{Lemma:LL-} is valid.
\end{proof}

\noindent From Lemmas \ref{Lemma:L+}, \ref{Lemma:L-}, and \ref{Lemma:LL-}, $LL^-(x_n)$ and $LL^+(x_n)$ are defined to prove that conditions \textit{(I)}, \textit{(II)}, and \textit{(III)} are satisfied:
\begin{definition}
	{A subset}, $LL^{-}(x_n)$, of $L^-(x_n)$ is defined {as} follows:
	\begin{equation}
		\hspace*{-2cm}
		LL^{-}(x_n) \equiv \left\lbrace (a_1(n), a_2(n))^T; 
		\bm{a}_n \in LL_1^-(x_n), 
		\prod_{i=k}^{n}J(x_i)\bm{a}_n \in LL_1^-(x_{k+1}), {}^\forall k \geq n\right\rbrace \label{LL-}.
	\end{equation}
\end{definition}

\noindent From Lemma \ref{Lemma:LL-}, we obtain the following: 
\begin{equation}
	~{}^\forall \bm{a}_n \in LL^-(x_n), J(x_n)\bm{a}_n \in LL^-(x_{n+1}). \label{LL--}
\end{equation}

\begin{definition}
	{The subset} $LL^+(x_n)$ of $L^+(x_n)$ is defined {as}
	\begin{equation}
		LL^+(x_n) 	\equiv
		\left\lbrace 
		\begin{array}{ll}
			\left\lbrace (a_1(n), a_2(n))^T :  -1<\frac{a_1(n)}{a_2(n)}<0\right\rbrace  & 
			\mbox{for}~\varepsilon> \frac{2}{\pi(\Delta \tau)^2},\\
			\left\lbrace (a_1(n), a_2(n))^T :  0<\frac{a_1(n)}{a_2(n)} < 1 \right\rbrace  & 
			\mbox{for}~\varepsilon<0.
		\end{array}
		\right. \label{LL+}
	\end{equation}
\end{definition}

\noindent 	Based on a straightforward calculation, the following equation is satisfied:
\begin{eqnarray}
	{}^\forall\bm{a}_n \in LL^+(x_n), J(x_n)\bm{a}_n \in LL^+(x_{n+1}). \label{LL++}
\end{eqnarray}

\begin{lemma}\label{Condition I and II}
	Conditions \textit{(I)} and \textit{(II)} 
  are satisfied when the condition $\varepsilon< 0~\mbox{or}~ \frac{2}{\pi(\Delta \tau)^2}<\varepsilon$ is satisfied.
\end{lemma}

\begin{proof}
	The eigenvector spaces $E_{x_n}^u$ and $E_{x_n}^s$ can be selected from $LL^+(x_n)$ and $LL^-(x_n)$, respectively{. Hence,} a combination of (\ref{LL-}) and (\ref{LL+}) results in:	\begin{eqnarray}
		E_{x_n}^u \subset LL^+(x_n),	~~E_{x_n}^s \subset LL^-(x_n). \label{eigenspace subset}
	\end{eqnarray}
	A combination of (\ref{LL--}), (\ref{LL++}), and (\ref{eigenspace subset}), results in	$J(x_n)  E_{x_n}^u \subset LL^+(x_{n+1})$ and 
	$J(x_n)  E_{x_n}^s \subset LL^-(x_{n+1})$. 
	Therefore, $E_{x_{n+1}}^u$ {and} $E_{x_{n+1}}^s$ can be {inductively selected} as follows:
	\begin{eqnarray}
		E_{x_{n+1}}^u \equiv J(x_n)   E_{x_n}^u, ~~
		E_{x_{n+1}}^s \equiv J(x_n)   E_{x_n}^s.
	\end{eqnarray}
	{The subspaces} $E_{x_n}^u$ and $E_{x_n}^s$ are linearly independent by definition, from which  
	
	\begin{eqnarray*}
		T_{x_n}M = E_{x_n}^u\oplus E_{x_n}^s.
	\end{eqnarray*}
\end{proof}

\begin{lemma}\label{Condition III}
	The condition \textit{(III)} is satisfied when the condition $\varepsilon< 0, \frac{2}{\pi(\Delta \tau)^2}<\varepsilon$ is satisfied. 
\end{lemma}

\begin{proof}
	We define $\alpha_n \equiv \left(1-\frac{\pi\varepsilon(\Delta \tau)^2}{\cos^2(\pi t_n)}\right)$ and
	the stretching rate $\sigma(x_n, \bm{a}_n) \equiv \frac{\|J(x_n)\bm{a}_n\|^2}{\|\bm{a}_n\|^2}$, which can be expressed as follows:
	\begin{eqnarray*}
		\sigma(x_n, \bm{a}_n)&=&\frac{a_1^2+a_2^2 -4\alpha_na_1a_2+4\alpha_n^2a_2^2}{a_1^2+a_2^2},\\
		&=& 1 -4\alpha_n\frac{a_1a_2}{a_1^2+a_2^2}
		+4\alpha_n^2\frac{a_2^2}{a_1^2+a_2^2},\\
		&=& 1+4\alpha_n\left( \alpha_n\sin^2\phi_n-\sin\phi_n\cos\phi_n\right),
	\end{eqnarray*}
	where $\phi_n$ is defined such that $\sin^2\phi_n = \frac{a_2^2(n)}{a_1^2(n)+a_2^2(n)}$ and 
	$\sin\phi_n\cos\phi_n =\frac{a_1(n)a_2(n)}{a_1^2(n)+a_2^2(n)},~-\pi< \phi\leq\pi$.
	Therefore, when $\bm{a}_n\in LL^+(x_n)$, the following equation is obtained:
	\begin{equation}
		F(\phi_n) \equiv \alpha_n\sin^2\phi_n-\sin\phi_n\cos\phi_n
		\left\lbrace 
		\begin{array}{cc}
			<0, &\varepsilon>\frac{2}{\pi(\Delta \tau)^2},\\
			>0, &\varepsilon<0.
		\end{array}\right.
	\end{equation}
	
	Two cases are worth considering for $\varepsilon$. The minimum value of $\sigma$ for each case can be determined as follows:
	
	\noindent(i) $\varepsilon> \frac{2}{\pi(\Delta \tau)^2}$,
	
	assuming that 
	$\bm{a}_n\in LL^+(x_n)\Leftrightarrow -1 < \frac{a_1(n)}{a_2(n)}= \cot\phi_n<0$, where $\phi_n$ is defined in the range 
	\begin{eqnarray}
		-\frac{\pi}{2} < \phi_n < -\frac{\pi}{4}, \mbox{or}~\frac{\pi}{2} < \phi_n < \frac{3}{4}\pi.
	\end{eqnarray}
	Because $F'(\phi)=\sin(2\phi)\left(\alpha_n - \cot(2\phi)\right)$ is positive in this range, the following equation is valid:
	\begin{equation}
	\begin{array}{rll}
	F(\phi_n) &<& F\left(-\frac{\pi}{4}\right)=F\left(\frac{3}{4}\pi\right) = \frac{1}{2}(\alpha_n +1)<0,\\
	\therefore \sigma(x_n, \bm{a}_n) &>& 1 +2\alpha_n(\alpha_n+1)\\
	&\geq& 1+ 2\left\lbrace 1-\pi\varepsilon(\Delta\tau)^2\right\rbrace \left\lbrace 2-\pi\varepsilon(\Delta\tau)^2\right\rbrace .
	\end{array}
	\end{equation}
	Therefore, by selecting $\displaystyle K = 1$, 
	$\frac{1}{\eta} = \sqrt{1+ 2\left\lbrace 1-\pi\varepsilon(\Delta\tau)^2\right\rbrace \left\lbrace 2-\pi\varepsilon(\Delta\tau)^2\right\rbrace }>1$, indicating that condition (\ref{stretching}) is satisfied.

	\noindent(ii) Case of $\varepsilon <0$,
	
	For $0< \cot\phi_n < 1$,
	the domain of $\phi_n$ is expressed as follows:
	\begin{eqnarray}
		-\frac{3}{4}\pi < \phi_n < -\frac{\pi}{2}, \mbox{or}~ \frac{\pi}{4} < \phi_n < \frac{\pi}{2}.
	\end{eqnarray}
	Because $F'(\phi)$ is positive in this range, the following equation can be obtained:  
	\begin{equation}
	\begin{array}{rll}
	F(\phi_n) &> & F\left(-\frac{3}{4}\pi\right) = F\left(\frac{\pi}{4}\right) = \frac{1}{2}(\alpha_n -1)>0,\\
	\therefore	\sigma(x_n, \bm{a}_n) &>& 1 +2\alpha_n(\alpha_n-1)\\
	&\geq& 1-2\pi\varepsilon(\Delta\tau)^2\left\lbrace 1-\pi\varepsilon(\Delta\tau)^2\right\rbrace .
	\end{array}
	\end{equation}
	Therefore, by selecting $\displaystyle K = 1$, 
	$\frac{1}{\eta} = \sqrt{1-2\pi\varepsilon(\Delta\tau)^2\left\lbrace 1-\pi\varepsilon(\Delta\tau)^2\right\rbrace }>1$, indicating that condition (\ref{stretching}) is satisfied.
	Furthermore, the relationship regarding $\eta$ ($\eta \to 1-0$ as $\varepsilon \to -0$) implies that the Lyapunov exponent converges to zero as $\varepsilon \to -0$.
	
	The stretching rate is defined as follows:
	$\sigma'(x_n, \bm{a}_n) \equiv \frac{\|J(\widetilde{T}_{\varepsilon,\Delta \tau}^{-1}x_n)^{-1}\bm{a}_n\|^2}{\|\bm{a}_n\|^2}
	=\frac{\|J(x_{n-1})^{-1}\bm{a}_n\|^2}{\|\bm{a}_n\|^2}$,
	which can be expressed as follows:
	\begin{eqnarray*}
		\sigma'(x_n, \bm{a}_n)&=&\frac{a_1^2+a_2^2 -4\alpha_{n-1}a_1a_2+4\alpha_{n-1}^2a_1^2}{a_1^2+a_2^2},\\
		&=& 1 -4\alpha_{n-1}\frac{a_1a_2}{a_1^2+a_2^2}
		+4\alpha_{n-1}^2\frac{a_1^2}{a_1^2+a_2^2},\\
		&=& 1+4\alpha_{n-1}\left( \alpha_{n-1}\sin^2\phi_n'-\sin\phi_n'\cos\phi_n'\right),
	\end{eqnarray*}
	where $\phi_n'$ is defined such that $\sin^2\phi_n' = \frac{a_1^2(n)}{a_1^2(n)+a_2^2(n)}$ and 
	$\sin\phi_n'\cos\phi_n' =\frac{a_1(n)a_2(n)}{a_1^2(n)+a_2^2(n)},~-\pi< \phi'\leq\pi$.
	Therefore, if $\bm{a}_n\in LL^-(x_n)$, the following equation is valid:
	\begin{equation}
		G(\phi_n') \equiv\alpha_{n-1}\sin^2\phi_n'-\sin\phi_n'\cos\phi_n'
		\left\lbrace 
		\begin{array}{cc}
			<0, &\varepsilon>\frac{2}{\pi(\Delta \tau)^2},\\
			>0, &\varepsilon<0.
		\end{array}\right.
	\end{equation}
	
	Two cases should be considered for $\varepsilon$. The minimum value of $\sigma'$ for each case is calculated as follows.
	
	\noindent(i) Case of $\varepsilon> \frac{2}{\pi(\Delta \tau)^2}$:
	
	Assuming 
	$\bm{a}_n\in LL^-(x_n)\Leftrightarrow -1 < \frac{a_2(n)}{a_1(n)}= \cot\phi_n'<0$, $\phi_n'$ is defined in the range 
	\begin{eqnarray}
		-\frac{\pi}{2} < \phi_n' < -\frac{\pi}{4}, \mbox{or}~\frac{\pi}{2} < \phi_n' <\frac{3}{4}\pi.
	\end{eqnarray}
	Similar to the aforementioned analysis, the following equation is satisfied:
	\begin{equation}
	\begin{array}{rll}
	G(\phi_n') &<&  \frac{1}{2}(\alpha_{n-1} +1)<0,\\
	\therefore \sigma'(x_n, \bm{a}_n) &>& 1 +2\alpha_{n-1}(\alpha_{n-1}+1)\\
	&\geq& 1+ 2\left\lbrace 1-\pi\varepsilon(\Delta\tau)^2\right\rbrace \left\lbrace 2-\pi\varepsilon(\Delta\tau)^2\right\rbrace .
	\end{array}
	\end{equation}
	Therefore, by selecting $\displaystyle K = 1$, 
	$\frac{1}{\eta} = \sqrt{1+ 2\left\lbrace 1-\pi\varepsilon(\Delta\tau)^2\right\rbrace \left\lbrace 2-\pi\varepsilon(\Delta\tau)^2\right\rbrace }>1$, the condition (\ref{shrinking}) is satisfied.
	
	\noindent(ii) Case of $\varepsilon <0$:
	
	For $0< \cot\phi_n' < 1$, the domain of $\phi_n'$ is defined as:
	\begin{eqnarray}
		-\frac{3}{4}\pi < \phi_n' < -\frac{\pi}{2}, \mbox{or}~ \frac{\pi}{4} < \phi_n' < \frac{\pi}{2}.
	\end{eqnarray}
	Therefore, similar to the aforementioned analysis, the following equation is satisfied:
	\begin{equation}
	\begin{array}{rll}
	G(\phi_n') &> &  \frac{1}{2}(\alpha_{n-1} -1)>0,\\
	\therefore	\sigma'(x_n, \bm{a}_n) &>& 1 +2\alpha_{n-1}(\alpha_{n-1}-1)\\
	&\geq& 1-2\pi\varepsilon(\Delta\tau)^2\left\lbrace 1-\pi\varepsilon(\Delta\tau)^2\right\rbrace .
	\end{array}
	\end{equation}
	Therefore, by selecting $\displaystyle K = 1$, 
	$\frac{1}{\eta} = \sqrt{1-2\pi\varepsilon(\Delta\tau)^2\left\lbrace 1-\pi\varepsilon(\Delta\tau)^2\right\rbrace }>1$,
	condition (\ref{shrinking}) is satisfied.
\end{proof}

\noindent Therefore, we present the first main theorem in this study.
\setcounter{theorem}{0}
\renewcommand{\thetheorem}{\Alph{theorem}}
\begin{theorem}\label{The:1}
	Suppose the condition $\varepsilon< 0~\mbox{or}~ \frac{2}{\pi(\Delta \tau)^2}<\varepsilon$ is satisfied.
	The map $\widetilde{T}_{\varepsilon, \Delta \tau}$ on $M$ is an Anosov diffeomorphism.
\end{theorem}
\setcounter{theorem}{15}
\renewcommand{\thetheorem}{\arabic{theorem}}
\begin{proof}
	From Lemma \ref{Condition I and II} and Lemma \ref{Condition III},
	the map $\widetilde{T}_{\varepsilon, \Delta \tau}$ on $M$ is an Anosov diffeomorphism.
\end{proof}

\section{Statistical properties when the linear instability condition is satisfied}
\label{sec: Statistical properties when the local instability condition is satisfied}
Here, we present some statistical properties in the dynamical system $(M, \widetilde{T}_{\varepsilon, \Delta \tau},\mu)$ where $\mu$ represents the Lebesgue measure. These properties include the uniqueness of the equilibrium state, uniqueness of the SRB measure, and existence of a physical measure that satisfies Pesin's formula, all of which are validated by the Anosov property.

In addition to Theorem A, based on Theorem 20.4.1. in \cite{katok1995introduction}, when a map $f: M\to M$ is an Anosov diffeomorphism and $M$ is a compact connected smooth Riemannian manifold if $f$ possesses a smooth invariant measure, this measure is equivalent to the equilibrium state. Furthermore, the dynamical system exhibits both mixing and topologically mixing properties. Therefore, the following corollary is valid:

\begin{corollary}\label{Cor:1}
Suppose the condition $\varepsilon< 0, \frac{2}{\pi(\Delta \tau)^2}<\varepsilon$ is satisfied. In that case, for the dynamical system, $(M, \widetilde{T}_{\varepsilon, \Delta \tau}, \mu)$ with $\mu$ representing the Lebesgue measure, $\mu$ serves as the equilibrium state. Additionally, the dynamical system $(M, \widetilde{T}_{\varepsilon, \Delta \tau}, \mu)$ demonstrates both mixing and topologically mixing properties.
\end{corollary}

The second theorem can be derived from Corollary \ref{Cor:1} as follows:
\setcounter{theorem}{1}
\renewcommand{\thetheorem}{\Alph{theorem}}
\begin{theorem}\label{The:2}
	Suppose the condition $\varepsilon< 0~\mbox{or}~ \frac{2}{\pi(\Delta \tau)^2}<\varepsilon$ is satisfied. In that case, the pair $(M, \widetilde{T}_{\varepsilon, \Delta \tau})$ conserves the Lebesgue measure as the unique equilibrium state. Furthermore, the dynamical system $(M, \widetilde{T}_{\varepsilon, \Delta \tau},\mu)$ is considered Bernoulli.
\end{theorem}
\setcounter{theorem}{16}
\renewcommand{\thetheorem}{\arabic{theorem}}
\begin{proof}
    Based on Corollary \ref{Cor:1}, the dynamical system $(M, \widetilde{T}_{\varepsilon, \Delta \tau},\mu)$ is considered ergodic if $\varepsilon< 0~\mbox{or}~ \frac{2}{\pi(\Delta \tau)^2}<\varepsilon$. Furthermore, based on \cite{arnold1968ergodic}, the dynamical system $(M, \widetilde{T}_{\varepsilon, \Delta \tau},\mu)$ is indecomposable. 
    Therefore, the non-wandering set of the map $\widetilde{T}_{\varepsilon, \Delta \tau}$ is defined as $NW(\widetilde{T}_{\varepsilon, \Delta \tau})$. Based on Proposition 18.6.5. in \cite{katok1995introduction}, 
    $NW(\widetilde{T}_{\varepsilon, \Delta \tau})=\mathbb{T}^2=M$ was established. Therefore, by employing spectral decomposition \cite{bowen1975equilibrium,brin2002introduction}, the non-wandering set can be represented by a unique basic set $\Omega_1$, such that $NW(\widetilde{T}_{\varepsilon, \Delta \tau})=M=\Omega_1$. Anosov's Closing Lemma \cite{bowen1975equilibrium} validates that $\widetilde{T}_{\varepsilon, \Delta \tau}$ is an Axiom A diffeomorphism. Hence, from Theorem 4.1. in \cite{bowen1975equilibrium}, the pair $(M, \widetilde{T}_{\varepsilon, \Delta \tau})$ posseses a unique equilibrium state. Furthermore, from Corollary \ref{Cor:1}, the Lebesgue measure $\mu$ serves as an equilibrium state. Therefore, if the condition $\varepsilon< 0~\mbox{or}~ \frac{2}{\pi(\Delta \tau)^2}<\varepsilon$ is satisfied for the dynamical system $(M, \widetilde{T}_{\varepsilon, \Delta \tau},\mu)$ with $\mu$ denoting the Lebesgue measure, $\mu$ is the unique equilibrium state. In addition, from Corollary \ref{Cor:1} and Theorem 4.1. in \cite{bowen1975equilibrium}, the dynamical system $(M, \widetilde{T}_{\varepsilon, \Delta \tau},\mu)$ is considered Bernoulli.
\end{proof}

\setcounter{theorem}{2}
\renewcommand{\thetheorem}{\Alph{theorem}}
\begin{theorem}\label{The:3}
    For the the dynamical system, $(M, \widetilde{T}_{\varepsilon, \Delta \tau}, \mu)$ where $\mu$ {represents} the Lebesgue measure, Pesin's formula is applicable.
\end{theorem}
\setcounter{theorem}{16}
\renewcommand{\thetheorem}{\arabic{theorem}}
\begin{proof}
    According to \cite{pesin1977characteristic,barreira2007nonuniform}, if a map $f$ defined on a compact Riemannian manifold is a $C^{1+\alpha}$ diffeomorphism and preserves a measure that is absolutely continuous with respect to the Lebesgue measure, then Pesin's formula is valid. Given that the map $\widetilde{T}_{\varepsilon, \Delta \tau}$ is a $C^{1+\alpha}$ map and the invariant measure $\mu$ represents the Lebesgue measure, Pesin's formula is valid for any $\varepsilon$.
\end{proof}

According to \cite{ledrappier1985metric,barreira2007nonuniform}, a hyperbolic measure $\nu$ which is invariant under a $C^{1+\alpha}$ diffeomorphism is an SRB measure if and only if it satisfies the Pesin's formula. Therefore, from Theorem C the following Corollary can be derived:
\begin{corollary}\label{Cor:2}
    Suppose condition $\varepsilon< 0, \frac{2}{\pi(\Delta \tau)^2}<\varepsilon$ is satisfied. In that case, for the dynamical system, $(M, \widetilde{T}_{\varepsilon, \Delta \tau}, \mu)$ with $\mu$ as the Lebesgue measure, $\mu$ is the SRB measure.
\end{corollary}

The third main theorem in this study is introduced as follows:
\setcounter{theorem}{3}
\renewcommand{\thetheorem}{\Alph{theorem}}
\begin{theorem}\label{The:4}
    Suppose the condition $\varepsilon< 0, \frac{2}{\pi(\Delta \tau)^2}<\varepsilon$ is satisfied. In that case, the pair $(M, \widetilde{T}_{\varepsilon, \Delta \tau})$ conserves the Lebesgue measure $\mu$ as the unique SRB measure.
\end{theorem}
\renewcommand{\thetheorem}{\arabic{theorem}}
\begin{proof}
    From Corollary \ref{Cor:1} and Corollary \ref{Cor:2}, the dynamical system $(M, \widetilde{T}_{\varepsilon, \Delta \tau}, \mu)$ exhibits topological transitivity and preserves the Lebesgue measure $\mu$ as the SRB measure for $\varepsilon<0, \frac{2}{\pi(\Delta\tau)^2}<\varepsilon$. Hence, as per Theorem 1.7 from \cite{rodriguez2011uniqueness}, the pair $(M, \widetilde{T}_{\varepsilon, \Delta \tau})$ conserves the Lebesgue measure as the unique SRB measure. 
\end{proof}

The fourth main theorem in this study is introduced as follows:
\setcounter{theorem}{4}
\renewcommand{\thetheorem}{\Alph{theorem}}
\begin{theorem}\label{The:5}
	Suppose the condition $\varepsilon< 0~\mbox{or}~ \frac{2}{\pi(\Delta \tau)^2}<\varepsilon$ is satisfied. In that case, for the dynamical system $(M, \widetilde{T}_{\varepsilon, \Delta \tau},\mu)$ with $\mu$ serving as the Lebesgue measure, $\mu$ is considered a physical measure.
\end{theorem}
\setcounter{theorem}{17}
\renewcommand{\thetheorem}{\arabic{theorem}}
\begin{proof}
    From Theorem A, D, and Corollary \ref{Cor:1}, the pair $(M, \widetilde{T}_{\varepsilon, \Delta \tau})$ conserves the Lebesgue measure as the ergodic SRB measure and does not possess a zero Lyapunov exponent when the condition $\varepsilon< 0~\mbox{or}~ \frac{2}{\pi(\Delta \tau)^2}<\varepsilon$ is satisfied. According to Theorem 3 from \cite{young2002srb}, the ergodic SRB measure $\mu$ serves as the physical measure.
\end{proof}

According to Corollary \ref{Cor:1}, the dynamical system defined as the triplet $(M, \widetilde{T}_{\varepsilon, \Delta \tau},\mu)$ demonstrates the mixing property, with $\mu$ representing the Lebesgue measure. Furthermore, Theorem B, D, and E established the convergence to the uniform distribution derived from the Lebesgue measure, which serves as the unique equilibrium state, unique SRB measure, and physical measure for $(M, \widetilde{T}_{\varepsilon, \Delta \tau},\mu)$. This convergence occurs for any  initial density function defined except for a zero volume set \cite{gallavotti2014nonequilibrium} when the condition (\ref{Condition}) is satisfied. Moreover, the map $\widetilde{T}_{\varepsilon, \Delta \tau}$ demonstrates time-reversal symmetry, ensuring convergence to the uniform distribution.

\section{Positivity of KS entropy and Lyapunov exponent}
\label{sec: Positivity of KS entropy and Lyapunov exponent}

In this section, we initially establish the positivity of the Kolmogorov-Sinai entropy when linear instability condition is satisfied. Subsequently, we validate the theoretical expectation of the largest Lyapunov exponent through numerical simulations.

From Theorem C, the Kolmogorov--Sinai (KS) entropy $h_{{KS}}(\widetilde{T}_{\varepsilon, \Delta \tau})$, representing the average amount of information per unit time, can be derived from Pesin’s formula \cite{barreira2007nonuniform,gaspard1990transport}. This entropy is calculated for almost all initial points as
\begin{eqnarray}
h_{{KS}}(\widetilde{T}_{\varepsilon, \Delta \tau}) &=& \int_M \log  \left\Vert D\widetilde{T}_{\varepsilon, \Delta \tau}|_{E^u}\right\Vert  \mu(dx),\label{KS entropy}
\end{eqnarray}
where $h_{{KS}}(\widetilde{T}_{\varepsilon, \Delta \tau})$ is positive
because $\widetilde{T}_{\varepsilon, \Delta\tau}$ is an Anosov diffeomorphism. 

According to Kac \cite{kac1959probability} and Chrikov \cite{chirikov2001big},  relaxation does not always occur monotonically, even when the system demonstrates mixing properties. However, under certain conditions, {the} KS entropy can be consistent with the time derivative of the entropy by assuming appropriate conditions \cite{latora1999kolmogorov,zaslavsky1985chaos,falcioni2007initial}.

Subsequently, we examine scenarios in which $N$ is sufficiently large to define $I_{\delta_{N,\Delta\tau,\varepsilon}}$ as $\left(-\frac{1}{2}, \frac{1}{2}\right]$ and where condition (\ref{Condition}) is satisfied.

To calculate $h_{{KS}}(\widetilde{T}_{\varepsilon, \Delta \tau})$, the eigenvalue associated with the
unstable direction, expressed by $\gamma(t), t \in I_{\delta_{N,\Delta\tau,\varepsilon}}$ is required with
\begin{eqnarray}
\gamma(t) &=& \left\lbrace 
\begin{array}{cl}
\gamma_-(t),&\mbox{when}~\varepsilon > \frac{2}{\pi(\Delta \tau)^2},\\
\gamma_+(t),&\mbox{when}~\varepsilon < 0,
\end{array}
\right.\\
\gamma_{\pm}(t) &=& 1- \frac{\pi\varepsilon(\Delta \tau)^2}{\cos^2(\pi {t})}\pm 
\sqrt{\left(1-\frac{\pi\varepsilon(\Delta \tau)^2}{\cos^2(\pi {t})}\right)^2-1}, \label{Eigenvalue}
\end{eqnarray}
where $h_{{KS}}(\widetilde{T}_{\varepsilon, \Delta \tau})$ is equivalent to the positive Lyapunov exponent $\lambda(\varepsilon, \Delta \tau)$
{owing} to the Pesin identity.
To obtain an explicit formula of the Lyapunov exponent, 
(\ref{Eigenvalue}) was substituted into $\left\Vert D\widetilde{T}_{\varepsilon, \Delta \tau}|_{E^u}\right\Vert$.

(i) {In} the case of $\varepsilon> \frac{2}{\pi(\Delta\tau)^2}$, the Lyapunov exponent is expressed as follows:
	\begin{equation}
	\hspace*{-2.5cm}
	\begin{array}{lllll}
	&&&&\lambda(\varepsilon, \Delta\tau) =\displaystyle \int_{-\frac{1}{2}}^{\frac{1}{2}}\log\left|
	1- \frac{\pi\varepsilon(\Delta \tau)^2}{\cos^2(\pi {t})}- 
	\sqrt{\left(1-\frac{\pi\varepsilon(\Delta \tau)^2}{\cos^2(\pi {t})}\right)^2-1}\right|dt\\
	&=&&&\displaystyle \int_{-\frac{1}{2}}^{\frac{1}{2}}\log\left|
	\cos^2(\pi t)-\pi\varepsilon(\Delta\tau)^2-
	\sqrt{\left(\cos^2(\pi t)-\pi\varepsilon(\Delta\tau)^2\right)^2-\cos^4(\pi t)}
	\right| dt\\
	&&&-&
	\displaystyle \int_{-\frac{1}{2}}^{\frac{1}{2}}\log\left|
	\cos^2(\pi t)
	\right| dt\\
	&=&&&\displaystyle \int_{-\frac{1}{2}}^{\frac{1}{2}}\log\left|
	\pi\varepsilon(\Delta\tau)^2-\cos^2(\pi t)+
	\sqrt{\left(\cos^2(\pi t)-\pi\varepsilon(\Delta\tau)^2\right)^2-\cos^4(\pi t)}
	\right| dt\\
	&&&-&
	\displaystyle\int_{-\frac{1}{2}}^{\frac{1}{2}}\log\left|
	\cos^2(\pi t)
	\right| dt\\
	&=&&& \displaystyle \int_{-\frac{1}{2}}^{\frac{1}{2}}\log\left|
	\cos^2(\pi t)-\pi\varepsilon(\Delta\tau)^2
	\right| dt
	-
	\int_{-\frac{1}{2}}^{\frac{1}{2}}\log\left|
	\cos^2(\pi t)
	\right| dt\\
	&&&+&
	\displaystyle\int_{-\frac{1}{2}}^{\frac{1}{2}}\log\left(
	1+
	\sqrt{1-\left(1-\frac{\pi\varepsilon(\Delta\tau)^2}{\cos^2(\pi t)}\right)^{-2}}
	\right) dt\\
	&=&&&\displaystyle \log\left[
	2\pi\varepsilon(\Delta\tau)^2-1+
	2\sqrt{\pi\varepsilon(\Delta\tau)^2\left\lbrace\pi\varepsilon(\Delta\tau)^2-1 \right\rbrace }
	\right]\\
	&&&+&
	\displaystyle\int_{-\frac{1}{2}}^{\frac{1}{2}}\log\left(
	1+
	\sqrt{1-\left(1-\frac{\pi\varepsilon(\Delta\tau)^2}{\cos^2(\pi t)}\right)^{-2}}
	\right) dt
	\end{array}
	\label{Eq:Lyapunov positive}
	\end{equation}

(ii) {In} the case of $\varepsilon < 0$, the Lyapunov exponent is expressed as follows:
	\begin{equation}
	\hspace*{-2.5cm}
	\begin{array}{lllll}
	&&&&\lambda(\varepsilon, \Delta\tau) =\displaystyle \int_{-\frac{1}{2}}^{\frac{1}{2}}\log\left|
	1- \frac{\pi\varepsilon(\Delta \tau)^2}{\cos^2(\pi {t})}+ 
	\sqrt{\left(1-\frac{\pi\varepsilon(\Delta \tau)^2}{\cos^2(\pi {t})}\right)^2-1}\right|dt\\
	&=&&&\displaystyle \int_{-\frac{1}{2}}^{\frac{1}{2}}\log\left|
	\cos^2(\pi t)-\pi\varepsilon(\Delta\tau)^2+
	\sqrt{\left(\cos^2(\pi t)-\pi\varepsilon(\Delta\tau)^2\right)^2-\cos^4(\pi t)}
	\right| dt\\
	&&&-&
	\displaystyle\int_{-\frac{1}{2}}^{\frac{1}{2}}\log\left|
	\cos^2(\pi t)
	\right| dt\\
	&=&&& \displaystyle \int_{-\frac{1}{2}}^{\frac{1}{2}}\log\left|
	\cos^2(\pi t)-\pi\varepsilon(\Delta\tau)^2
	\right| dt
	-
	\int_{-\frac{1}{2}}^{\frac{1}{2}}\log\left|
	\cos^2(\pi t)
	\right| dt\\
	&&&+&
	\displaystyle\int_{-\frac{1}{2}}^{\frac{1}{2}}\log\left(
	1+
	\sqrt{1-\left(1-\frac{\pi\varepsilon(\Delta\tau)^2}{\cos^2(\pi t)}\right)^{-2}}
	\right) dt\\
	&=&&& \displaystyle \log\left[
	1-2\pi\varepsilon(\Delta\tau)^2+2\sqrt{\pi\varepsilon(\Delta\tau)^2\left\lbrace \pi\varepsilon(\Delta\tau)^2-1\right\rbrace }
	\right]\\
	&&&+&
	\displaystyle\int_{-\frac{1}{2}}^{\frac{1}{2}}\log\left(
	1+
	\sqrt{1-\left(1-\frac{\pi\varepsilon(\Delta\tau)^2}{\cos^2(\pi t)}\right)^{-2}}
	\right) dt
	\end{array}
	\label{Eq:Lyapunov negative}
	\end{equation}
From Equation (\ref{Eq:Lyapunov negative}), when $|\varepsilon|\ll 1$, the following equation is valid:
\begin{equation}
\begin{array}{lll}
&&\log\left[1-2\pi\varepsilon(\Delta\tau)^2+2\sqrt{\pi\varepsilon(\Delta\tau)^2\left\lbrace \pi\varepsilon(\Delta\tau)^2-1\right\rbrace }\right]\\
&=&O\left(
-2\pi\varepsilon(\Delta\tau)^2+2\sqrt{\pi\varepsilon(\Delta\tau)^2\left\lbrace \pi\varepsilon(\Delta\tau)^2-1\right\rbrace }
\right)\\
&=& O(|\varepsilon|^\frac{1}{2}). 
\end{array} \label{Eq: Lyapunov critical exponent 1}
\end{equation}
In terms of $\displaystyle \int_{-\frac{1}{2}}^{\frac{1}{2}}\log\left(
1+\sqrt{1-\left(1-\frac{\pi\varepsilon(\Delta\tau)^2}{\cos^2(\pi t)}\right)^{-2}}
\right) dt$, the following equation is valid:
\begin{equation}
\begin{array}{lll}
&&\log\left(
1+\sqrt{1-\left(1-\frac{\pi\varepsilon(\Delta\tau)^2}{\cos^2(\pi t)}\right)^{-2}}
\right),\\
&=& O\left(\sqrt{1-\left(1-\frac{\pi\varepsilon(\Delta\tau)^2}{\cos^2(\pi t)}\right)^{-2}}\right)\\
&=& O(|\varepsilon|^\frac{1}{2}). 
\end{array} \label{Eq: Lyapunov critical exponent 2}
\end{equation}
Therefore, from Equations (\ref{Eq: Lyapunov critical exponent 1}) and (\ref{Eq: Lyapunov critical exponent 2}),
when $0<-\varepsilon \ll 1$, then
\begin{equation}
\lambda(\varepsilon, \Delta\tau) = O(|\varepsilon|^\frac{1}{2}).
\end{equation}
Therefore, the critical exponent of the Lyapunov exponent is $\frac{1}{2}$ as $\varepsilon \to -0$. This critical exponent aligns with those found in intermittent chaos. \cite{pomeau1980intermittent,umeno2016exact,okubo2018universality,okubo2022universal}.

Given the complexities of expressions (\ref{Eq:Lyapunov positive}) and (\ref{Eq:Lyapunov negative}), we {made certain assumptions} to facilitate obtaining an explicit analytical expression.  

\noindent (Assumption): 
{The} equivalent value of (\ref{KS entropy}) can be obtained provided 
$\left\langle \cos^2(\pi t)\right\rangle = 1/2$ into $\cos^2(\pi t)$ can be substituted into  Equations (\ref{Eq:Lyapunov positive}) and (\ref{Eq:Lyapunov negative}). 

{Building upon} this assumption,
Equations (\ref{Eq:Lyapunov positive}) and (\ref{Eq:Lyapunov negative}) can be integrated as follows:
{\footnotesize
	\begin{equation}
	\hspace*{-2cm}
	\lambda(\varepsilon, \Delta \tau) \simeq  
	\left\lbrace
	\begin{array}{cl}
	\log\left[2\pi\varepsilon(\Delta \tau)^2-1+4\sqrt{\pi\varepsilon(\Delta \tau)^2\left\lbrace \pi\varepsilon(\Delta \tau)^2-1\right\rbrace}+
	\frac{4\pi\varepsilon(\Delta \tau)^2\left\lbrace \pi\varepsilon(\Delta \tau)^2-1\right\rbrace }{2\pi\varepsilon(\Delta \tau)^2-1}\right],& 
	\varepsilon> \frac{2}{\pi(\Delta \tau)^2},\\
	\log\left[1-2\pi\varepsilon(\Delta \tau)^2+4\sqrt{\pi\varepsilon(\Delta \tau)^2\left\lbrace \pi\varepsilon(\Delta \tau)^2-1\right\rbrace}+
	\frac{4\pi\varepsilon(\Delta \tau)^2\left\lbrace \pi\varepsilon(\Delta \tau)^2-1\right\rbrace }{1-2\pi\varepsilon(\Delta \tau)^2}\right], &
	\varepsilon<0.
	\end{array}
	\right.
	\end{equation}
}
These equations also represent the theoretical expectation of the Kolmogorov-Sinai entropy.
The comparison between the numerical result and analytical formula of the Lyapunov exponent is shown in Figure \ref{Fig: Lyapunov exponent}.
The numerical result is consistent with the analytical formula.
\begin{figure}[h!]
	\centering
	\includegraphics[width=.7\columnwidth]{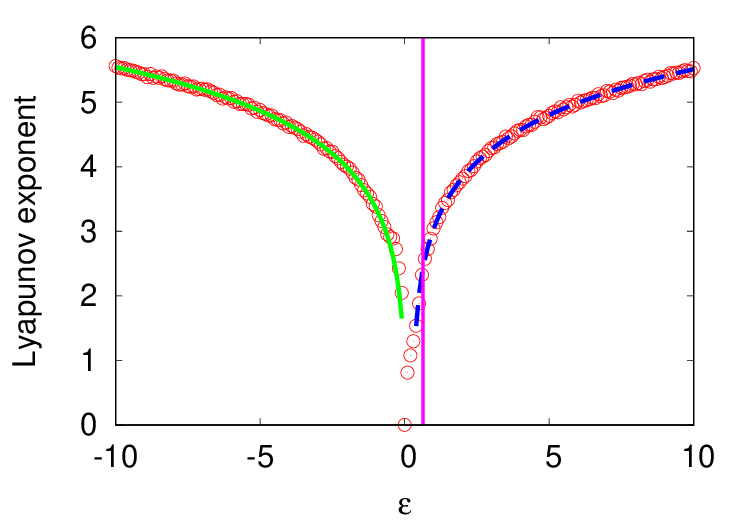}
	\caption{Numerical calculation of the Lyapunov exponent (red circle) and analytical formula (lines) when $\Delta \tau =1$.
		{The} solid line corresponds to $\lambda(\varepsilon, \Delta \tau)$ for $\varepsilon< 0$ and {the} broken line
		corresponds to $\varepsilon> \frac{2}{\pi}$. The iteration number is $10^4$, with the initial deviation vector
		and point specified as follows:
		$v_{s}=\frac{1}{\sqrt{2}}, v_{t}= \frac{1}{\sqrt{2}}$ and 
		$(s,t) = (\frac{\sqrt{2}-1}{2}, \frac{\sqrt{3}-1}{2})$. 
		{The} vertical line corresponds to $\varepsilon=\frac{2}{\pi}$.}
	\label{Fig: Lyapunov exponent}
\end{figure}

\section{Conclusion}
\label{sec: Conclusion}
This study demonstrated that the symplectic map $\widetilde{T}_{\varepsilon,\Delta\tau}$, derived from a certain Hamiltonian, functions as an Anosov diffeomorphism when the condition of linear instability was satisfied at the microscopic level. 
At the macroscopic level, we established the convergence of the initial density function, defined within the dynamical system excluding the volume-zero region, toward a uniform distribution. This uniform distribution corresponded to a unique equilibrium state, unique SRB measure, and physical measure. 
This observation suggests the coexistence of time-reversal symmetry at the microscopic level, derived from the Hamiltonian, and irreversibility at the macroscopic level.

The properties of $\widetilde{T}_{\varepsilon,\Delta\tau}$ derived from this study indicate that for $\varepsilon >0$, $\widetilde{T}_{\varepsilon,\Delta\tau}$ exhibited  the Anosov property
as the time step $\Delta\tau$ increased. Further, we demonstrated that it possessed the Anosov property for arbitrary time steps in the region where $\varepsilon$ is negative, implying that $\widetilde{T}_{\varepsilon,\Delta\tau}$ remained an Anosov diffeomorphism even as the time step $\Delta\tau$ approached infinitesimally small values, thereby validating the Anosov property in the continuous system.

The question arises: what implications does this result have on the mixing (irreversibility) of the original mapping $T_{\varepsilon, \Delta \tau}$ with time-reversal symmetry? 
Two possibilities emerge. First, both $T_{\varepsilon,\Delta\tau}$ and $\widehat{T}_{\varepsilon,\Delta\tau}$ exhibit a mixing property, thereby implying that this property can be inferred from $\widetilde{T}_{\varepsilon,\Delta\tau}$. Second, both 
$T_{\varepsilon,\Delta\tau}$ and $\widehat{T}_{\varepsilon,\Delta\tau}$ lack a mixing property, and the mixing property of $\widetilde{T}_{\varepsilon,\Delta\tau}$ results from the compactification of its phase space (folding). The resolution of this theoretical question remains a subject for future investigation.

We propose two potential research directions:

First, we suggest investigating diffusive phenomena within the $T_{\varepsilon,\Delta\tau}, \widehat{T}_{\varepsilon,\Delta\tau}$ mapping system. 
The leapfrog method, known for its symplectic numerical integration method that preserves time-reversal symmetry, ensures that $T_{\varepsilon,\Delta\tau}, \widehat{T}_{\varepsilon,\Delta\tau}$ mappings also preserve time-reversal symmetry with respect to momentum inversion. 
Regarding the time-reversal symmetry in the $\widetilde{T}_{\varepsilon,\Delta\tau}$ mapping in this study, 
we utilized a mathematically extended definition of physical time-reversal symmetry, focusing specifically on momentum, as we believe this symmetry regarding momentum inversion is crucial owing to the presence of physical observables. 
While this study did not demonstrate mixing (macroscopic irreversibility) in $T_{\varepsilon,\Delta\tau}, \widehat{T}_{\varepsilon,\Delta\tau}$, 
future research is expected to delve into mixing and relaxation phenomena in systems where time-reversal symmetry regarding momentum reversal is present. 
Particularly, within the $T_{\varepsilon,\Delta\tau}, \widehat{T}_{\varepsilon,\Delta\tau}$ mappings, relaxation of the initial density function to a uniform distribution on $\mathbb{R}$ is anticipated through the sum of random variables possessing the mixing property and adhering to the Cauchy distribution. 

Second, we propose studying a model derived from a Hamiltonian with numerous degrees of freedom proven to exhibit the mixing property. This research aims to investigate the statistical properties 
of systems that satisfy the assumptions of Gallavotti's chaos hypothesis. 
A multi-particle system can be constructed by preparing various ensembles for our $\widetilde{T}_{\varepsilon,\Delta\tau}$ mapping. 
However, this multi-particle system lacks interactions. In a multi-particle system with interactions, phenomena not observed in a system devoid of interactions may become apparent.

\section*{Acknowledgment}
The authors express their gratitude to Prof. Masayuki Asaoka and Dr. Yoshiyuki Y. Yamaguchi for their invaluable advice and informative discussions. Ken-ichi Okubo acknowledges the support received from the Grant-in-Aid for JSPS Research Fellow Grant Number JP17J07694 and Grant-in-Aid for Early-Career Scientists Number 23K16963.


 \bibliographystyle{elsarticle-num} 
 \bibliography{cas-refs}





\end{document}